\documentclass[11pt]{amsart}
\usepackage[margin=2.5cm]{geometry}
\usepackage{enumerate}
\usepackage{amsmath}
\usepackage{amssymb,latexsym}
\usepackage{amsthm}
\usepackage{color}
\usepackage[dvipsnames]{xcolor}
\usepackage{cancel}
\usepackage{graphicx}
\usepackage{cite}
\usepackage{changes}
\makeatletter
\renewcommand{\tocsection}[3]{%
  \indentlabel{\@ifnotempty{#2}{\bfseries\ignorespaces#1 #2\quad}}\bfseries#3}
\renewcommand{\tocsubsection}[3]{%
  \indentlabel{\@ifnotempty{#2}{\ignorespaces#1 #2\quad}}#3}

\usepackage{etoolbox}
\makeatletter
\patchcmd{\@setaddresses}{\indent}{\noindent}{}{}
\patchcmd{\@setaddresses}{\indent}{\noindent}{}{}
\patchcmd{\@setaddresses}{\indent}{\noindent}{}{}
\patchcmd{\@setaddresses}{\indent}{\noindent}{}{}
\makeatother

\usepackage[linktocpage]{hyperref}
\hypersetup{
  colorlinks   = true, 
  urlcolor     = blue, 
  linkcolor    = Purple, 
  citecolor   = red 
}
\usepackage{comment}
\usepackage{amscd}
\usepackage{mathtools}

\DeclareMathOperator{\C}{\mathcal{C}}
\newcommand{\srk}{\mathrm{srk}}
\newcommand{\dsrk}{\mathrm{d}_{\mathrm{srk}}}

\DeclareMathOperator{\supp}{supp}
\DeclareMathOperator{\rk}{rk}

\DeclareMathOperator{\dd}{d}
\DeclareMathOperator{\Mat}{Mat}

\DeclareMathOperator{\colsp}{colsp}

\DeclareMathOperator{\GL}{GL}

\DeclareMathOperator{\ww}{w}

\theoremstyle{definition}
\newtheorem{theorem}{Theorem}[section]

\newtheorem{lemma}[theorem]{Lemma}
\newtheorem{corollary}[theorem]{Corollary}
\newtheorem{definition}[theorem]{Definition}
\newtheorem{proposition}[theorem]{Proposition}

\newtheorem{example}[theorem]{Example}

\newtheorem{notation}[theorem]{Notation}
\newtheorem{remark}[theorem]{Remark}

\newcommand{\cC}{{\mathcal C}}

\newcommand{\F}{{\mathbb F}}

\newcommand{\NN}{{\mathbb N}}

\newcommand{\bfn}{\mathbf {n}}
\newcommand{\bfm}{\mathbf {m}}

\newcommand{\fq}{{\mathbb F}_{q}}
\newcommand{\Fq}{{\mathbb F}_{q}}
\newcommand{\Fm}{{\mathbb F}_{q^m}}
\newcommand{\fqm}{{\mathbb F}_{q^m}}

\newcommand{\la}{\langle}
\newcommand{\ra}{\rangle}

\newcommand{\PG}{\mathrm{PG}}

\newcommand{\spac}{\Mat(\bfn,\bfm,\Fq)}
\newcommand{\Fmnk}{[\bfn,k]_{q^m/q}}
\newcommand{\Fmnkd}{[\bfn,k,d]_{q^m/q}}
\newcommand{\Fmk}{[n,k]_{q^m/q}}

\newcommand{\st}{\,:\,}

\title{One-weight codes in the sum-rank metric}
\date{}

\author[U. Mushrraf]{Usman Mushrraf}
\address{Usman Mushrraf, \textnormal{Dipartimento di Matematica e Fisica, Universit\`a degli Studi della Campania ``Luigi Vanvitelli'', Viale Lincoln, 5, I--\,81100 Caserta, Italy}}
\email{usman.mushrraf@unicampania.it}

\author[F. Zullo]{Ferdinando Zullo}
\address{Ferdinando Zullo, \textnormal{Dipartimento di Matematica e Fisica, Universit\`a degli Studi della Campania ``Luigi Vanvitelli'', Viale Lincoln, 5, I--\,81100 Caserta, Italy}}
\email{ferdinando.zullo@unicampania.it}

\begin{document}

\begin{abstract}
One-weight codes, in which all nonzero codewords share the same weight, form a highly structured class of linear codes with deep connections to finite geometry. While their classification is well understood in the Hamming and rank metrics—being equivalent to (direct sums of) simplex codes—the sum-rank metric presents a far more intricate landscape. In this work, we explore the geometry of one-weight sum-rank metric codes, focusing on three distinct classes.
First, we introduce and classify \emph{constant rank-list} sum-rank codes, where each nonzero codeword has the same tuple of ranks, extending results from the rank-metric setting. Next, we investigate the more general \emph{constant rank-profile} codes, where, up to reordering, each nonzero codeword has the same tuple of ranks. Although a complete classification remains elusive, we present the first examples and partial structural results for this class. Finally, we consider one-weight codes that are also MSRD (Maximum Sum-Rank Distance) codes. For dimension two, constructions arise from partitions of scattered linear sets on projective lines. For dimension three, we connect their existence to that of special $2$-fold blocking sets in the projective plane, leading to new bounds and nonexistence results over certain fields.
\end{abstract}

\maketitle
\tableofcontents

\section{Introduction}

In the study of linear codes over finite fields with the Hamming metric, a one-weight code is a linear code in which all nonzero codewords have the same Hamming weight. These codes are rare and highly structured, and their classification reveals deep connections between coding theory and finite geometry.
The simplest example of a one-weight code is the repetition code, where all nonzero codewords are scalar multiples of a single vector of full support. A classical family of one-weight codes is given by the \emph{simplex codes} (also known as MacDonald codes \cite{macdonald1960design}), which are the duals of Hamming codes. Over $\mathbb{F}_q$, the simplex code has parameters $\left[\frac{q^m - 1}{q - 1},\ m,\ q^{m-1}\right]$, for some $m \in \mathbb{N}$, and all nonzero codewords have Hamming weight $q^{m-1}$.

Linear one-weight codes in the Hamming metric have been classified as precisely those codes that are equivalent to direct sums of simplex codes. This result has been established in several papers using a variety of techniques; see \cite{bonisoli1984every,peterson1962error,assmus1963error,weiss1966linear}.

The notion of one-weight codes naturally extends to the rank metric. In \cite{randrianarisoa2020geometric} (see also \cite{alfarano2022linear}), Randrianarisoa proved a rank-metric analogue of the classification for linear one-weight codes in the Hamming metric. Specifically, he showed that all linear one-weight rank-metric codes are equivalent to the simplex rank-metric code, i.e., a code with parameters $[mk, k, m]_{q^m/q}$.

Surprisingly, this classification does not extend to the sum-rank metric. In \cite{neri2023geometry}, the authors presented classes of sum-rank metric codes of dimension two that are not equivalent to the simplex sum-rank code. Later, in \cite{borello2024geometric}, the authors found examples of one-weight sum-rank metric codes of arbitrary dimension that are also not equivalent to the simplex sum-rank code, demonstrating that the variety of one-weight sum-rank codes is significantly richer.

These results suggest that obtaining a classification of one-weight codes in the sum-rank metric is likely to be a very challenging task.

In this paper, using a geometric perspective, we provide new insights into the theory of one-weight sum-rank metric codes. We focus on two classes of such codes.
First, we consider sum-rank metric codes in which every nonzero codeword has the same list of ranks—that is, the same vector of local ranks across the blocks. We refer to these as \emph{constant rank-list} sum-rank metric codes. For this class, we provide a classification result, building on the work of Randrianarisoa in \cite{randrianarisoa2020geometric}.

Next, we study codes in which every nonzero codeword has the same rank profile, meaning the same multiset of local ranks (i.e., the rank-list up to reordering). We call these \emph{constant rank-profile} sum-rank metric codes. While a full classification remains out of reach, we present the first known examples of such codes and give a characterization in certain cases. Here again, the geometric approach plays a central role.

The final class of one-weight codes we study consists of those that are also MSRD codes, i.e., codes that are optimal with respect to the Singleton bound. As we point out, constructing one-weight MSRD codes of dimension two can be obtained via a partition of scattered linear sets on the projective lines. However, the case of dimension three is already far from trivial, as also noted in \cite{martinez2023doubly}.
By employing counting arguments, we relate the existence of such codes to the existence of a particular type of $2$-fold blocking set in the projective plane. This connection allows us to derive bounds on the number of blocks in a putative one-weight MSRD code of dimension three, as well as to prove some nonexistence results in the case $q=2$ and for field extensions of degree two.

\subsection{Organization of the paper}
The structure of the paper is as follows. Section 2 introduces the preliminaries on linear sets. In Section 3, we review the basics of sum-rank metric codes and their geometric interpretation. Section 4 is devoted to a special class of one-weight codes based on a constant rank-list in the context of the sum-rank setting. In Section 5, we present another family of one-weight codes characterized by a constant rank-profile and explore their properties. Finally, in Section 6, we explore one-weight MSRD codes with particular attention to the cases where the dimension is two or three.  

\section{Preliminaries on linear sets}
In this paper, we use the geometric perspective of sum-rank metric codes, which is closely related to the theory of linear sets. Therefore, we will quickly review the definition and some basic properties of linear sets.
Let $\mathbb{V}$ be a $k$-dimensional $\fqm$-vector space and let $\Lambda=\PG(\mathbb{V},\fqm)=\PG(k-1,q^m)$. 
Clearly, $\mathbb{V}$ can also be seen as both an $\fqm$-vector space of dimension $k$ and as an $\fq$-vector space of dimension $mk$. We can consider an $\fq$-subspace $U$ of $\mathbb{V}$ and define the set
\[ L_U=\{\langle {u} \rangle_{\fqm} \colon {u}\in U\setminus\{{0}\}\} \]
as an $\fq$-\textbf{linear set} of rank $\dim_{\fq}(U)$.
The \textbf{weight} of a projective subspace $S=\PG(W,\fqm) \subseteq \Lambda$ in $L_U$ is defined as $w_{L_U}(S)=\dim_{\fq}(U \cap W)$.

Let us review some fundamental relationships involving the size of a linear set, the count of points with a specific weight, and its rank.
If $L_U$ has rank $n$ then the weight of any point is bounded by $n$. Denote by $N_i(U)$ the number of points of $\Lambda$ having weight $i\in \{0,\ldots,n\}$ in $L_U$, the following relations hold:
\begin{equation}\label{eq:card}
    |L_U| \leq \frac{q^n-1}{q-1},
\end{equation}
\begin{equation}\label{eq:pesicard}
    |L_U| =N_1(U)+\ldots+N_n(U),
\end{equation}
\begin{equation}\label{eq:pesivett}
    N_1(U)+N_2(U)(q+1)+\ldots+N_n(U)(q^{n-1}+\ldots+q+1)=q^{n-1}+\ldots+q+1.
\end{equation}

For further details we refer to \cite{polverino2010linear}.

\section{Sum-rank metric codes and their geometry}

 We fix now the notation that we will use for the whole paper. For us $q$ is a prime power and $\Fq$ is the finite field with $q$ elements. We will often consider the degree $m$ extension field $\Fm$ of $\Fq$. 
\noindent Let $t$ be a positive integer. From now on
$\bfn=(n_1,\ldots,n_t), \bfm=(m_1,\ldots,m_t) \in \NN^t$ will always be ordered tuples with $n_1 \geq n_2 \geq \ldots \geq n_t$ and $m_1 \geq m_2, \geq \ldots \geq m_t$, and we set $N\coloneqq n_1+\ldots+n_t$. We use the following compact notations for the direct sum of vector spaces 
$$  \Fm^\bfn\coloneqq\bigoplus_{i=1}^t\Fm^{n_i} \,\, \text{and}\,\,\spac\coloneqq \bigoplus_{i=1}^t \F_q^{n_i \times m_i}.$$

For a vector $\mathbf{a}=(a_1,\ldots, a_r)\in \NN^r$, we define $$\mathcal{S}_{\mathbf{a}}=\mathcal{S}_{a_{1}} \times \ldots \times \mathcal{S}_{a_r},$$
where $\mathcal{S}_i$ is the symmetric group of order $i$. Similarly, we denote by $\GL(\mathbf{a}, \F_q)$ the direct product of the general linear groups of degree $a_i$ over $\F_q$, i.e.
$$ \GL(\mathbf{a}, \F_q) = \GL(a_1, \F_q)\times \ldots\times \GL(a_r, \F_q).$$

We will first recall the framework of matrices for sum-rank metric codes and later the vectorial framework.

\subsection{Matrix codes}

In this section we recall the basic notions of sum-rank-metric codes seen as elements in $\spac$, that will be useful for the rest of the paper. The interested reader is referred to \cite{byrne2021fundamental,byrne2020anticodes,moreno2021optimal} for a more detailed description of this setting. 

\begin{definition}
Let $X\coloneqq(X_1,\dots, X_t)\in\Mat(\mathbf{n},\mathbf{m},\Fq)$. 
The \textbf{sum-rank support} of $X$ is defined as the space
$$\supp(X)\coloneqq(\colsp(X_1), \colsp(X_2),\ldots, \colsp(X_t)) \subseteq \Fq^\bfn,$$
where $\colsp(X_i)$ is the $\F_q$-span of the columns of $X_i$.
The \textbf{rank-list} of $X$ is defined as 
$$\rho(X) \coloneqq (\dim(\colsp(X_1)),\dots, \dim(\colsp(X_t)))\in\mathbb{N}^t.$$
The \textbf{rank-profile} of $X$ is the $t$-uple obtained from $\rho(X)$ by rearranging its  entries in non-increasing order and it is denoted by $\mu(X)$. Finally, the \textbf{sum-rank weight} of $X$ is the quantity
$$\ww_{\srk}(X)\coloneqq\dim_{\Fq}(\supp(X))=\sum_{i=1}^t \rk(X_i).$$
\end{definition}

With these definition in mind, we can endow the space $\spac$ with a distance function, called the \textbf{sum-rank distance},
\[
\dsrk : \spac \times \spac \longrightarrow \mathbb{N}  
\]
defined by
\[
\dsrk(X,Y) \coloneqq \ww_{\srk}(X-Y).
\]

Therefore, we can now give the definition of sum-rank metric codes.

\begin{definition}
 A \textbf{sum-rank metric code} $\mathcal{C}$ is an $\F_q$-linear subspace of $\spac$ endowed with the sum-rank distance. 
The \textbf{minimum sum-rank distance} of a sum-rank code $\mathcal{C}$ is defined as usual via $$\dsrk(\mathcal{C})\coloneqq\min\{\ww_{\srk}(X) \st X \in \mathcal{C}, X \neq \mathbf{0}\}.$$ 
The \textbf{sum-rank support} of the code $\mathcal{C}$ is the $\Fq$-span of the supports of all the codewords of $\C$, that is
$$ \supp(\C)\coloneqq\sum_{X\in\C} \supp(X) \subseteq \Fq^\bfn.$$
Finally,  we say that $\mathcal{C}$ is \textbf{sum-rank  nondegenerate} if $\supp(\mathcal{C})=\F_q^\bfn$.
\end{definition}

For the nondegeneracy notion see also \cite{santonastaso2025invariants}.

\subsection{Vector codes} \label{sub:vectorsetting}

Let us focus on the scenario where $m=m_1=\cdots=m_t$. We consider an alternate setting for the study of sum-rank metric codes, where codewords are vectors with entries from an extension field $\F_{q^m}$ instead of matrices over $\F_q$. For a comprehensive explanation of this, see \cite{martinezpenas2018skew,martinez2019universal,martinez2019theory,neri2021twisted,ott2021bounds,martinez2022codes}.

The \textbf{$\Fq$-rank} of a vector $v=(v_1,\ldots,v_n) \in \F_{q^m}^n$ is the $\Fq$-dimension of the
vector space generated over $\F_q$ by its entries, i.e, 
$$\rk_q(v)\coloneqq\dim_{\fq} (\langle v_1,\ldots, v_n\rangle_{\fq}).$$ 
Let $x=(x_1 , \ldots,  x_t)\in\F_{q^m}^\bfn$, with $x_i\in\F_{q^m}^{n_i}$ for any $i$. We can now extend the notion of rank to a $t$-ple of vectors, defining the \textbf{sum-rank weight of $x$} as the quantity
$$ \ww(x)=\sum_{i=1}^t \rk_q(x_i).$$
Consequently, the 
\textbf{sum-rank distance on $\Fm^\bfn$}, is  
\[
\dd(x,y)=\sum_{i=1}^t \rk_q(x_i-y_i),
\]
for any  $x=(x_1 , \ldots, x_t), y=(y_1, \ldots, y_t) \in \F_{q^m}^\bfn$, with $x_i,y_i \in \F_{q^m}^{n_i}$. 

Therefore, a sum-rank metric code can be defined in this framework as follows.

\begin{definition}
Let $k$ be a positive integer with $1 \le k \le N$. A \textbf{sum-rank metric code} $\C$ is a $k$-dimensional $\Fm$-subspace of $\Fm^\bfn$ endowed with the sum-rank metric. 
The \textbf{minimum sum-rank distance} of $\C$ is the integer
\[
\dd(\C)=\min\{\dd(x,y) \st x, y \in \C, x\neq y  \}= \min\{\ww(x) \st x \in \C, x\neq 0  \}.
\]
We will write that $\C$ is an $\Fmnkd$ code if $k$ is the $\Fm$-dimension of $\C$ and $d$ is its minimum distance, or simply an $\Fmnk$ code if the minimum distance is not relevant/known.
\end{definition}

Since the elements of $\Fm^\bfn$ can also be seen as long vectors in $\Fm^N$, every $\Fmnk$ code can be provided with generator and parity-check matrices. Each of them, is naturally partitioned as $G=(G_1 \mid \ldots \mid G_t)$, where $G_i \in \Fm^{k\times n_i}$ (respectively $H=(H_1 \mid \ldots \mid H_t)$, where $H_i \in \Fm^{(N-k)\times n_i}$).

\medskip 

For vector sum-rank metric codes, a Singleton-like bound holds.
\begin{theorem}[see \textnormal{\cite[Proposition 16]{martinezpenas2018skew}}]  \label{th:Singletonbound}
     Let $\mathcal{C}$ be an $\Fmnkd$ code. Then 
     \[
     d \leq N-k+1.
     \]
\end{theorem}

An $\Fmnkd$ code is called a \textbf{Maximum Sum-Rank Distance code} (or shortly \textbf{MSRD code}) if $d=N-k+1$.

The next result establishes the $\Fm$-linear isometries of the space $\F_{q^m}^\bfn$ endowed with the sum-rank distance, and it has been proved in \cite[Theorem 3.7]{alfarano2021sum}. The case $\bfn=(n,\ldots, n)$ was already proved in \cite[Theorem 2]{martinezpenas2021hamming}.
\begin{theorem}
The group of $\F_{q^m}$-linear isometries of the space $(\F_{q^m}^\bfn,\dd)$  is
$$((\F_{q^m}^\ast)^{t} \times \GL(\bfn, \F_q)) \rtimes \mathcal{S}_{\lambda(\bfn)},$$
which (right)-acts as 
  \begin{equation*} (x_1 , \ldots, x_t)\cdot (\mathbf{a},A_1,\ldots, A_t,\pi) \longmapsto (a_1x_{\pi(1)}A_1 \mid \ldots \mid a_tx_{\pi(t)} A_{t}).\end{equation*}
\end{theorem}

Given that MacWilliams's extension theorem does not hold in this context (refer to \cite{barra2015macwilliams}), we will define the equivalence of sum-rank metric codes using $\Fm$-linear isometries of the entire ambient space.

\begin{definition}\label{def:equiv_codes} We say that two $\Fmnk$ sum-rank metric codes $\cC_1, \cC_2$ are \textbf{equivalent} if there is an $\Fm$-linear isometry $\phi$ such that $\phi(\cC_1)=\cC_2$. 
\end{definition}

We denote the set of equivalence classes of $\Fmnkd$ sum-rank metric codes by $\mathfrak{C}\Fmnkd$.

\subsection{Supports and degeneracy of sum-rank metric codes}

It is known that  to a vector sum-rank metric code we can associate a matrix sum-rank metric code with the same parameters and metric properties in the following way.

For every $r \in [t]$, let $\Gamma_r=(\gamma_1^{(r)},\ldots,\gamma_m^{(r)})$ be an ordered $\Fq$-basis of $\Fm$, and let $\Gamma\coloneqq(\Gamma_1,\ldots,\Gamma_t)$. Given   $x=(x_1, \ldots ,x_t) \in \Fm^\bfn$, with $x_i \in \Fm^{n_i}$, define the element $\Gamma(x)=(\Gamma_1(x_1), \ldots, \Gamma_t(x_t)) \in \spac$, where $\mathbf{m}=(m,\ldots,m)$ and 
$$x_{r,i} = \sum_{j=1}^m \Gamma_r (x_r)_{ij}\gamma_j^{(r)}, \qquad \mbox{ for all } i \in [n_r].$$
In other words, the $r$-th block of $\Gamma(x)$ is the matrix expansion of the vector $x_r$ with respect to the $\Fq$-basis $\Gamma_r$ of $\Fm$.

One can observe that the map 
$$\Gamma: \Fm^\bfn \longrightarrow \spac$$
is an $\Fq$-linear isometry between the metric spaces $(\Fm^\bfn, \dd)$ and $(\spac,\dsrk)$; see \cite[Theorem 2.7]{neri2023geometry}. 
This demonstrates that the two settings, namely the matrix framework and the vectorial one, are isometric. Consequently, we can operate within both contexts.

This isometry enables us to express support, rank profile, and rank-list within the vector framework; see \cite{martinez2019theory}.

\begin{definition}
The \textbf{sum-rank support} of an element $x=(x_1, \ldots, x_t) \in \Fm^{\bfn}$ is the tuple
$$\supp_{\bfn}(x)\coloneqq\supp(\Gamma(x)),$$
for any choice of $\Gamma=(\Gamma_1,\ldots, \Gamma_t)$, where $\Gamma_i$ is an $\Fq$-basis of $\Fm$ for each $i \in[t]$.
The \textbf{rank-list} of $x=(x_1, \ldots, x_t)$ is the $t$-uple $\rho(x)=(\rk_q(x_1), \ldots,\rk_q(x_t))$. The \textbf{rank-profile} of $x$ is the $t$-uple obtained from $\rho(x)$ by rearranging the entries in non-increasing order and it is denoted by $\mu(x)$.
\end{definition}

In the same way, we can extend the notions of sum-rank support and degeneracy for an $\Fmk$ code.
\begin{definition}
 We define the \textbf{sum-rank support of an $\Fmnk$ code} $\C$ as the $\F_q$-span of the supports of the codewords of $\C$ and we denote it by $\supp(\C)$, i.e.
$$ \supp_{\bfn}(\C)\coloneqq\sum_{c\in\C}\supp_{\bfn}(c).$$
Furthermore, we say  that $\C$ is \textbf{sum-rank  nondegenerate} if $\supp(\C)=\Fq^\bfn$. We say that $\C$ is \textbf{sum-rank degenerate} if it is not sum-rank  nondegenerate.
\end{definition}

\subsection{Geometry of sum-rank metric codes}\label{sec:3}

Hamming-metric codes correspond to projective systems \cite[Theorem 1.1.6]{tsfasman1991algebraic}, and rank-metric codes to $q$-systems, their $q$-analogues \cite{sheekey2019scatterd,randrianarisoa2020geometric,alfarano2022linear}. The geometric view of sum-rank metric codes has been recently pointed out in \cite{neri2023geometry} (see also \cite{santonastaso2023subspace}).
We are going to describe it, as it will be crucial in our study.
\medskip

\begin{theorem}\label{th:connection}(see \cite[Theorem 3.1]{neri2023geometry})
Let $\mathcal{C}$ be a  nondegenerate $\Fmnkd$ code with generator matrix $G=(G_1|\ldots|G_t)$.
Let $\mathcal{U}_i$ be the $\F_q$-span of the columns of $G_i$, for $i\in [t]$. Then, for every $v\in \Fm^k$ and $i \in [t]$ we have
\begin{equation}\label{eq:weight_dimension}
\ww(v G) = N - \sum_{i=1}^t \dim_{\fq}(\mathcal{U}_i \cap v^{\perp}),\end{equation}
where $v^{\perp}=\{y \in \F_{q^m}^k \colon v \cdot y=0\}$ and $v \cdot y$ denotes the standard scalar product between $v$ and $y$.
In particular,
\[
d=N - \max\left\{ \sum_{i =1}^t\dim_{\fq}(\mathcal{U}_i \cap H)  \colon H\mbox{ is an } \F_{q^m}\mbox{-hyperplane of }\F_{q^m}^k  \right\}.
\]
\end{theorem}

Building upon the last result, the following corollary establishes the link between the rank-list of codewords and the dimension-list of the system.

\begin{corollary}\label{cor:geometricranklist}
Let $\mathcal{C}$ be a  nondegenerate $\Fmnkd$ code with generator matrix $G = (G_1 \mid \dots \mid G_t)$. Let $\mathcal{U}_i$ be the $\mathbb{F}_q$-span of the columns of $G_i$, for $i \in [t]$. Then, for every $v \in \mathbb{F}_{q^m}^k$ and $i \in [t]$, we have
\[\rho (vG)=(\mathrm{rk}_q(vG_1),\ldots,\mathrm{rk}_q(vG_t))=(n_1-\dim_{\mathbb{F}_q}(\mathcal{U}_1\cap v^\perp),\ldots, n_t-\dim_{\mathbb{F}_q}(\mathcal{U}_t\cap v^\perp)).\]
\end{corollary}

The above results point out that the metric properties of a code can be completely translated into hyperplane intersection properties with certain subspaces. 

\begin{definition}
An \textbf{$\Fmnkd$ system} $\mathcal{U}$ is an ordered set $(\mathcal{U}_1,\ldots,\mathcal{U}_t)$, where, for any $i\in [t]$, $\mathcal{U}_i$ is an $\F_q$-subspace of $\F_{q^m}^k$ of dimension $n_i$, such that
$ \langle \mathcal{U}_1, \ldots, \mathcal{U}_t \rangle_{\F_{q^m}}=\F_{q^m}^k$ and 
$$ d=N-\max\left\{\sum_{i=1}^t\dim_{\F_q}(\mathcal{U}_i\cap \mathcal H) \st \mathcal H \textnormal{ is an $\F_{q^m}$-hyperplane of }\F_{q^m}^k\right\}.$$
\end{definition}

In the following we give the definition of equivalent systems.

\begin{definition}\label{def:equiv_systems}
Two $\Fmnkd$ systems $(\mathcal{U}_1,\ldots,\mathcal{U}_t)$ and $(\mathcal{V}_1,\ldots, \mathcal{V}_t)$ are \textbf{equivalent} if there exists an isomorphism $\varphi\in\GL(k,\F_{q^m})$, an element $\mathbf{a}=(a_1,\ldots,a_t)\in (\Fm^*)^t$ and a permutation $\sigma\in\mathcal{S}_t$, such that for every $i\in[t]$
$$ \varphi(\mathcal{U}_i) = a_i\mathcal{V}_{\sigma(i)}.$$
\end{definition}

We define $\mathfrak{U}\Fmnkd$ as the set of equivalence classes of $\Fmnkd$ systems.

Consider $[\C]\in\mathfrak{C}\Fmnkd$ where $\C$ is a  nondegenerate $\Fmnk$ code. Let $G=(G_1 \,\mid\, \ldots \,\mid\, G_t)$ be a generator matrix for $\cC$. Consider the equivalence class $[\mathcal{U}]$ of $\Fmnkd$ systems, where $\mathcal{U}=(\mathcal{U}_1,\ldots,\mathcal{U}_t)$ and $\mathcal{U}_i$ is the $\F_q$-span of the columns of $G_i$ for every $i\in[t]$. 
We will refer to $\mathcal{U}$ as a \textbf{system associated with } $\C$.
Viceversa, given $[(\mathcal{U}_1,\ldots,\mathcal{U}_t)]\in\mathfrak{U}\Fmnkd$, for every $i\in[t]$, fix an $\F_q$-basis $\{g_{i,1}, \ldots, g_{i,n_i}\}$ of $\mathcal{U}_i$. Define $G_i$ as the matrix whose columns are $\{g_{i,1}, \ldots, g_{i,n_i}\}$ and consider the equivalence class of the sum-rank metric code generated by $G=(G_1 \,\mid\, \ldots \,\mid\, G_t)$.
We refer to $\mathcal{C}$ as a \textbf{code associated with } $\mathcal{U}$.
Thanks to \cite[Theorem 3.7]{neri2023geometry}, we know that two correspondences define a one-to-one correspondence between $\mathfrak{C}\Fmnkd$ and $\mathfrak{U}\Fmnkd$.

\section{Constant rank-list sum-rank-metric codes}

In this section, we completely classify constant rank-list sum-rank-metric codes, proving a sum-rank metric analogue of the classification of one-weight codes shown in the Hamming metric in \cite{bonisoli1984every,peterson1962error,assmus1963error,weiss1966linear} and in the rank metric by Randrianarisoa in \cite{randrianarisoa2020geometric} (see also \cite{alfarano2022linear}).

A constant rank-list sum-rank-metric code is a sum-rank metric code in which all the nonzero codewords have the same rank-list.

\begin{definition}
    Let $\mathcal{C}$ be an $[\mathbf{n},k,d]_{q^m/q}$ code and let $(\rho_1,\ldots,\rho_t)\in \NN^t$. 
    If for any nonzero codeword $c \in \mathcal{C}$ we have $\rho(c)=(\rho_1,\ldots,\rho_t)$, then we say that $\mathcal{C}$ is a $(\rho_1,\ldots,\rho_t)$-\textbf{constant rank-list} sum-rank metric code.
\end{definition}

We will use the following notation.

\begin{notation}
    Let $c=(c_1,\ldots,c_t)\in\F_{q^m}^{\mathbf{n}}$ and let $i \in [t]$, the $i$-th component of $c$, denoted as $c_{|i}$, is defined as \[c_{|i}=c_i\in\mathbb{F}_{q^m}^{n_i} .\]
\end{notation}

We can now associate to a sum-rank metric code in $\F_{q^m}^{\mathbf{n}}$ some rank-metric codes. 

\begin{definition}
    Let $\mathcal{C}$ be an $[\mathbf{n},k,d]_{q^m/q}$ code.
    For each $i\in [t]$, we define the \textbf{$i$-th projection} of $\mathcal{C}$ as
    \[\mathcal{C}_i=\{c_{|i} \colon c\in\mathcal{C}\} \subseteq \F_{q^m}^{n_i}.\]
\end{definition}

In the sum-rank metric, a constant rank-list implies that the projections of the code are rank-metric codes with a single weight.

\begin{proposition}\label{prop:owc}
    Let $\mathcal{C}$ be an $[\mathbf{n},k,d]_{q^m/q}$ $(\rho_1,\ldots,\rho_t)$-constant rank-list sum-rank metric code. For any $i \in [t]$, the $i$-th projection $\mathcal{C}_i$ of $\mathcal{C}$ is a one-weight rank-metric code in $\F_{q^m}^{n_i}$ with minimum distance $\rho_i$.
\end{proposition}
\begin{proof}
    Suppose that $\mathcal{C}$ is a $(\rho_1,\ldots,\rho_t)$-constant rank-list sum-rank metric code. For any $c=(c_1,\ldots,c_t) \in \mathcal{C}$ we have
    \[ \rho(c)=(w(c_1),\ldots,w(c_t))=(\rho_1,\ldots,\rho_t), \]
    i.e. $w(c_i)=\rho_i$ for any $i \in[t]$.
    Therefore, the assertion is proved.
\end{proof}

A class of one-weight rank-metric codes is the following.

\begin{definition}
   A rank-metric code with parameters $[km,k,m]_{q^m/q}$ is called \textbf{simplex rank-metric code}. 
\end{definition}

A simplex rank-metric code is clearly a one-weight rank-metric code.
In \cite{randrianarisoa2020geometric} (see also \cite{alfarano2022linear}) the one-weight codes in the rank-metric have been classified.

\begin{theorem}[see \text{\cite[Corollary 3.17]{alfarano2022linear}}]\label{thm:onew}
     Let $k\geq 2$ and let $\mathcal{C}$ be an $[n,k,d]_{q^m/q}$ one-weight code. Then, $\mathcal{C}$ is isometric to a simplex rank-metric $[km,k,m]_{q^m/q}$ code.
 \end{theorem}

As a consequence of this result, we have also a classification for one-weight  nondegenerate rank-metric codes.

\begin{corollary}\label{cor:onew}
    Let $k\geq 2$ and let $\mathcal{C}$ be a  nondegenerate $[n,k,d]_{q^m/q}$ one-weight code. Then $\mathcal{C}$ is the simplex rank-metric $[km,k,m]_{q^m/q}$ code.
\end{corollary}

The next result discusses about sum-rank  nondegenerate codes and explains that the projections of  nondegenerate sum-rank metric codes are  nondegenerate codes in rank metric, as stated in the following proposition (see also \cite[Corollary 3.9]{santonastaso2025invariants}).

\begin{proposition}\label{prop:nd}
    Let $\mathcal{C}$ be a  nondegenerate $[\mathbf{n},k,d]_{q^m/q}$ sum-rank metric code. For any $i \in [t]$, the $i$-th  projection $\mathcal{C}_i$ of $\mathcal{C}$ is a  nondegenerate code in the rank metric.
\end{proposition}
\begin{proof}
     Let $i \in [t]$ and suppose that $\mathcal{C}_i$ is degenerate for some $i\in[t]$. 
     This means that $\text{supp}(\mathcal{C}_i)\subset \fq^{n_i}$. 
     Note that 
     \[\text{supp}_{\textbf{n}}(\mathcal{C})_i\subseteq \text{supp}(\mathcal{C}_i),\]
     since $\text{supp}(\mathcal{C}_i)\subset \fq^{n_i}$ we have that $\text{supp}_{\textbf{n}}(\mathcal{C})\ne \fq^{\textbf{n}}$.
     However, this contradicts our assumption that $\mathcal{C}$ is sum-rank  nondegenerate.  
\end{proof}

The following is the main result of this section, which provides a classification of constant rank-list sum-rank-metric codes. This statement serves as a sum-rank metric analogue to the well-established classifications of one-weight codes presented in the Hamming metric and in the rank metric.

\begin{theorem}\label{thm:classoneweightlist}
 Let $\mathcal{C}$ be a nondegenerate $[\mathbf{n},k,d]_{q^m/q}$ sum-rank metric code. If $\mathcal{C}$ is a $(\rho_1,\ldots,\rho_t)$-constant rank-list sum-rank metric code then 
 \begin{itemize}
     \item $\text{dim}(\mathcal{C}_i)=k$ for any $i \in [t]$;
     \item if $k\geq 2$, then we have that $\mathcal{C}_i$ is an $[mk,k,m]_{q^m/q}$ code;
     \item if $k=1$, then for any $i \in [t]$ $\mathcal{C}_i$ is an $[n_i,1,n_i]_{q^m/q}$ code;
     \item there are no nonzero codewords in $\C$ having a zero block.
 \end{itemize} 
Conversely, if for $\C$ the previous items hold, then $\C$ is a nondegenerate constant rank-list sum-rank metric code.
\end{theorem}
\begin{proof}
Suppose that $\mathcal{C}$ is a $(\rho_1,\ldots,\rho_t)$-constant rank-list sum-rank metric code.
By Proposition \ref{prop:nd}, the projections of $\mathcal{C}$ are nondegenerate rank-metric codes for all $i\in[t]$. Since each $\mathcal{C}_i$ is a  nondegenerate one-weight rank-metric code, we can apply Corollary \ref{cor:onew} to those codes for which $k_i=\dim(\mathcal{C}_i)\geq2$. For such projections, the length of each $\mathcal{C}_i$ is $mk_i$ and their minimum distance is $m$. 
If $\dim(\mathcal{C}_i)=1$, then trivially, the parameters of $\mathcal{C}_i$ are $[n_i,1,n_i]_{q^m/q}$.
Finally, $\dim(\mathcal{C}_i)$ cannot be zero as the code $\mathcal{C}$ is a  nondegenerate sum-rank metric code.
It cannot happen that there exists two indices $i_1,i_2 \in [t]$ such that $k_{i_1}\ne k_{i_2}$, otherwise it would exist a nonzero codeword $c$ with $c_{i_1}=0$ or $c_{i_2}=0$, implying that one of the $\rho_i$'s is zero, a contradiction to the nondegenerate condition. 
Clearly, there are no nonzero codewords in $\C$ having a zero block otherwise we would get a contradiction to the fact that every nonzero codeword has constant rank-list.
The converse also holds because of the fact that the $\C_i$'s are one-weight code.
\end{proof}
Based on the previous result, we can show some restrictions on the parameters of a  nondegenerate sum-rank code with constant rank-list.
\begin{corollary}
    Let $\mathcal{C}$ be a  nondegenerate $[\mathbf{n},k,d]_{q^m/q}$  $(\rho_1,\ldots,\rho_t)$-constant rank-list sum-rank metric code.
    For any $i \in[t]$, let $k=\dim(\mathcal{C}_i)$ for any $i \in [t]$.
    We have that
    \begin{itemize}
        \item if $k\geq 2$, then $n_i= mk$ for every $i$;
        \item $\rho_i=\begin{cases}
    m & \text{if } k\geq 2,\\
    n_i & \text{if } k=1.
    \end{cases}$
    \end{itemize}
    We have that the minimum distance of $\C$ is $tm$ if $k\geq 2$ or $\sum_{i \in [t]}n_i$ if $k=1$.
\end{corollary}
We conclude this section by showing an example of constant rank-list sum-rank metric codes.
\begin{example}
Let $\alpha$ be a primitive element of $\F_{q^m}$ and consider 
\begin{align*}
G&= \left(
\begin{array}{c|c|c|c}
\begin{matrix}
I_k & \alpha I_k & \cdots & \alpha^{m-1} I_k \\
\end{matrix}
&
\begin{matrix}
I_k & \alpha I_k & \cdots & \alpha^{m-1} I_k
\end{matrix}
&
\cdots
&
\begin{matrix}
I_k & \alpha I_k & \cdots & \alpha^{m-1} I_k
\end{matrix}
\end{array}
\right)\\
&=(G_1|\cdots|G_s) \in \F_{q^m}^{k\times smk},    
\end{align*}
where $I_k$ denotes the identity matrix in $\F_{q^m}^{k\times k}$.
Let $\C$ be the sum-rank metric code having $G$ as generator matrix and note that it is nondegenerate.
Note that, since $(I_k,\alpha I_k,\cdots,\alpha^{m-1} I_k)$ is the generator matrix of the simplex code with parameters $[mk,k,m]_{q^m/q}$ by \cite[Proposition 3.16]{alfarano2022linear}, we have that by Theorem \ref{thm:classoneweightlist}, $\C$ is a  nondegenerate $[(mk,\ldots,mk),k,sm]_{q^m/q}$ $(m,\ldots,m)$-constant rank-list sum-rank metric code.
\end{example}

\section{Constant rank-profile sum-rank metric codes}

In this section, we will study sum-rank metric codes whose nonzero weights have a constant rank-profile; the study of these codes began in \cite{neri2023geometry}. Although the situation appears very similar to that of constant rank-list sum-rank metric codes, we will demonstrate that examples of such codes exist, but a complete classification result seems quite challenging to achieve.

\begin{definition}
    Let $\mathcal{C}$ be an $[\bfn,k,d]_{q^m/q}$ sum-rank metric code and let $(\mu_1,\dots,\mu_t)\in \NN^t$ such that $\mu_1\geq \ldots \geq \mu_t$.
    We say that $\mathcal{C}$ is a $(\mu_1,\dots,\mu_t)$-\textbf{constant rank-profile} sum-rank metric code if for every nonzero codeword $c \in \mathcal{C}$ we have
    \[ \mu(c)=(\mu_1,\dots,\mu_t). \]
\end{definition}

A case study is one where the blocks of the sum-rank metric all have the same length. Indeed, the following necessary conditions have been provided for constant rank-profile sum-rank metric codes by leveraging the connection with the Hamming metric.

\begin{corollary} [see \text{\cite[Corollary 5.8]{neri2023geometry}}]\label{cor:parameters}
    Let $\mathcal{C}$ be a nondegenerate  $ [(n, \dots, n), k]_{q^m / q} $ $(\mu_1,\dots,\mu_t)$-constant rank-profile sum-rank metric code. Then
    \begin{itemize}
        \item[(1)] The number
        \[
        \ell = \frac{t(q^n - 1)(q^m - 1)}{(q - 1)(q^{km} - 1)}
        \]
        is a positive integer.
        \item[(2)] It holds that
        \[
        t q^{m(k-1)} (q^n - 1)(q^m - 1) = (q^{km} - 1) \left( t q^n - \sum_{i=1}^t q^{n - \mu_i} \right).
        \]
    \end{itemize}
\end{corollary}

As a first interesting case to study, the authors of \cite{neri2023geometry} considered the following.

\begin{example}\label{exa:constantrankprofile}(see \text{\cite[Example 5.9]{neri2023geometry}})
Let us assume that we want to construct a  nondegenerate  $[(n,\ldots,n),k]_{q^m/q}$ code with constant rank-profile $(\mu_1,\ldots,\mu_t)$, and let us fix $n=3$, $k=m=2$. 
Using the above corollary, one can get that the constant rank-profile must be
$$ ( \underbrace{2, \ldots, 2}_{t'q^2 \text{ {times}}}, \underbrace{1, \ldots, 1}_{t' \text{ {times}}}),$$
for some $t' \in \NN$.
\end{example}

Our aim is to use the geometric view to get some constructions of constant rank-profile sum-rank metric codes. To this aim, we need the following map.

\begin{notation}
    Define the following map
\[
\tau: \mathbb{N}^t \to \mathbb{N}^t
\]
such that for any \( (a_1, a_2, \dots, a_t) \in \mathbb{N}^t \), the image
\[
\tau(a_1, \dots, a_t) = (b_1, \dots, b_t)
\]
satisfies
\[
b_1 \geq b_2 \geq \dots \geq b_t,
\]
and \( (b_1, \dots, b_t) \) is a reordering of \( (a_1, \dots, a_t) \).
\end{notation}

We start by giving the geometric interpretation of constant rank-profile sum-rank metric codes.

\begin{proposition}\label{prop:geometricranklist}
Let $\mathcal{C}$ be a  nondegenerate $(\mu_1,\dots,\mu_t)$-constant rank-profile sum-rank metric code with parameters $[\bfn, k, d]_{q^m/q}$. Let $G = (G_1 \mid \dots \mid G_t)$ be a generator matrix of $\mathcal{C}$. Let $\mathcal{U}_i$ be the $\mathbb{F}_q$-span of the columns of $G_i$, for every $i \in [t]$. 
For all $v \in \mathbb{F}_{q^m}^k$, we have that
\[\tau(n_1-\dim_{\mathbb{F}_q}(\mathcal{U}_1\cap v^\perp),\ldots, n_t-\dim_{\mathbb{F}_q}(\mathcal{U}_t\cap v^\perp))= (\mu_1,\dots,\mu_t). \]
\end{proposition}
\begin{proof}
   By Corollary \ref{cor:geometricranklist}, we have the rank-list \[(\mathrm{rk}_q(vG_1),\ldots,\mathrm{rk}_q(vG_t))=(n_1-\dim_{\mathbb{F}_q}(\mathcal{U}_1\cap v^\perp),\ldots, n_t-\dim_{\mathbb{F}_q}(\mathcal{U}_t\cap v^\perp)).\]
   Since $\C$ is constant rank-profile, we have that 
   \[\tau(\mathrm{rk}_q(vG_1),\ldots,\mathrm{rk}_q(vG_t))=\tau(n_1-\dim_{\mathbb{F}_q}(\mathcal{U}_1\cap v^\perp),\ldots, n_t-\dim_{\mathbb{F}_q}(\mathcal{U}_t\cap v^\perp))=(\mu_1,\ldots,\mu_t),\]
   so that, we obtain the desired result.
\end{proof}

In order to provide a construction for Example \ref{exa:constantrankprofile}, in the next results of this section we will construct examples of constant rank-list in the case in which the code has dimension two.

\begin{remark}
    For the geometric point of view, we need to control the intersection of the systems with the hyperplane of $\F_{q^m}^2$, which are the one-dimensional $\F_{q^m}$-subspaces of $\F_{q^m}^2$.
    The number of one-dimensional $\F_{q^m}$-subspaces in $\F_{q^m}^2$ is $q^m+1$ and we can list them as follows
    \[ \Lambda_1=\{  \langle {x}_1 \rangle_{\F_{q^m}}, \dots, \langle {x}_{q^m+1} \rangle_{\F_{q^m}} \}, \]
    for some ${x}_1,\ldots,{x}_{q^m+1} \in \F_{q^m}^2$.
    We will fix this set (and the representatives of the one-dimensional subspaces) throughout this paper.
\end{remark}

To construct a family of  nondegenerate constant rank-profile sum-rank metric codes we will use the geometric description.

\begin{proposition}\label{prop:constrmmme}
    Let $(\mathcal{U}_1,\ldots,\mathcal{U}_t)$ be the $[\mathbf{n},2]_{q^m/q}$ system defined as follows
    \begin{itemize}
        \item $\mathbf{n}=(m+e,\ldots,m+e)$, for some $e \in \NN$ such that $0\leq e< m$;
        \item for any $i \in [q^m+1]$
    \[
    \mathcal{U}_{i} = \langle {x}_i \rangle_{\F_{q^m}} \oplus \mathcal{S}_{e}^i,
    \]
    where $ \mathcal{S}_{e}^i $ is an $ e $-dimensional $\fq$-subspace of $ \mathbb{F}_{q^m}^2 $;
    \item for any $i \in [t]$, $ \mathcal{S}_{e}^i \cap \langle {x}_i \rangle_{\F_{q^m}} = \{0\} $.
    \end{itemize}   
    Let $\mathcal{C}$ be a sum-rank metric code associated with the system $(\mathcal{U}_1,\ldots,\mathcal{U}_t)$. 
    We have that $\mathcal{C}$ is a  nondegenerate $({ m, \dots, m}, e)$-constant rank-profile sum-rank metric code with parameters $[\mathbf{n}, 2, q^mm+e]_{q^m/q}$.
\end{proposition}
\begin{proof}
    We start by observing that for any $i, j \in [q^m+1]$ we have
    \[ \dim_{\fq}(\mathcal{U}_i\cap \langle {x}_j\rangle_{\F_{q^m}})=
    \begin{cases}
        m & \text{if } i=j,\\
        e & \text{otherwise}.
    \end{cases}
    \]
    Since $\dim_{\fq}(\mathcal{U}_i)=m+e$ for any $i$, by Proposition \ref{prop:geometricranklist} we have that
    \[ \mu({x}G)=(m,\ldots,m,e), \]
    for any $x \in \F_{q^m}^2\setminus\{0\}$.
    This implies that $\mathcal{C}$ is a  nondegenerate $({ m, \dots, m}, e)$-constant rank-profile sum-rank metric code.
    The minimum distance is then the sum of the entries of the rank-list $({ m, \dots, m}, e)$.
\end{proof}

The above proposition shows a construction for Example \ref{exa:constantrankprofile} when choosing $k=m=2$ and $e=1$. In this case $t'=1$. In the next corollary, we extend it to the case of $t'>1$.

\begin{corollary}
    For any $t' \in \NN$ there exists a  nondegenerate $(\underbrace{ m, \dots, m}_{q^m t' \text{ times}}, \underbrace{ e, \dots, e}_{t' \text{ times}})$-constant rank-profile sum-rank metric code with parameters $[(m+e,\ldots,m+e), 2, t'q^mm+t'e]_{q^m/q}$.
\end{corollary}
\begin{proof}
    Let $\mathcal{C}'$ be as in Proposition \ref{prop:constrmmme}.
    Consider $\mathcal{C}$ as the repetition code of $\mathcal{C}'$ $t'$ times, i.e.
    \[ \C=\{ (\underbrace{ c, \dots, c}_{t' \text{ times}}) \colon c \in \C' \}. \]
    The assertion then follows.
\end{proof}

We can characterize the codes with the above parameters in the following way.

\begin{theorem}
Let $\mathcal{C}$ be a  nondegenerate $(\underbrace{ m, \dots, m}_{t' \text{ times}}, \underbrace{ e, \dots, e}_{t - t' \text{ times}})$-constant rank-profile sum-rank metric code with parameters $[\mathbf{n}, 2, d]_{q^m/q}$, where $0\leq e< m$ and $n_i<2m$ for any $i$.
We have that 
\begin{itemize}
    \item $ t = t'(q^m + 1) $;
    \item $n_i=m+e$ for any $i\in[t]$;
    \item a system associated with $\mathcal{C}$ is $(\mathcal{U}_{1,1},\ldots,\mathcal{U}_{q^m+1,t'})$ where
    \[
    \mathcal{U}_{i,j} = \langle {x}_i \rangle_{\F_{q^m}} \oplus \mathcal{S}_{e}^j,
    \]
    where $ \mathcal{S}_{e}^j $ is an $ e $-dimensional $\fq$-subspace of $ \mathbb{F}_{q^m}^2 $ with the property that $ \mathcal{S}_{e}^j \cap \langle {x}_i \rangle_{\F_{q^m}} = \{0\} $.
\end{itemize}
\end{theorem}
 \begin{proof}
Let $G = (G_1 \mid \dots \mid G_t)$ be a generator matrix of $\mathcal{C}$. Let $\mathcal{U}_i$ be the $\mathbb{F}_q$-span of the columns of $G_i$, for every $i \in [t]$. 
By assumption and by Corollary \ref{cor:geometricranklist} we have that for any $j \in [q^m+1]$
\begin{equation}\label{eq:t'andt} \dim_{\fq}(\mathcal{U}_i\cap \langle {x}_j \rangle_{\F_{q^m}})=
    \begin{cases}
        m & \text{for } t' \text{ values of } i , \\
        e & \text{for the remaining values of } i.
    \end{cases} \end{equation}

\textbf{Case 1}: Suppose that $\mathcal{U}_i$ is such that there exists $j \in[q^m+1]$ for which $ \dim_{\fq}(\mathcal{U}_i\cap \langle {x}_j \rangle_{\F_{q^m}})=m$. 
This implies that $\mathcal{U}_i\supseteq \langle {x}_j \rangle_{\F_{q^m}}$.
By assumption, we also have that
\[ \dim_{\fq}(\mathcal{U}_i\cap \langle {x}_h \rangle_{\F_{q^m}})\in \{e,m\}. \]
Note that it cannot happen that for some $h\ne j$ we have $ \dim_{\fq}(\mathcal{U}_i\cap \langle {x}_h \rangle_{\F_{q^m}})=m$, because otherwise $\mathcal{U}_i\supseteq \langle {x}_j,{x}_h \rangle_{\F_{q^m}}=\F_{q^m}^2$, a contradiction to $n_i< 2m$.
Therefore, we have that if $\mathcal{U}_i$ is such that $\mathcal{U}_i\supseteq \langle {x}_j \rangle_{\F_{q^m}}$ for some $j$ then
    \begin{equation}\label{eq:condUi}  
    \dim_{\fq}(\mathcal{U}_i\cap \langle {x}_h \rangle_{\F_{q^m}})=
    \begin{cases}
        m & \text{if } h=j, \\
        e & \text{otherwise.}
    \end{cases} 
    \end{equation}
As a consequence, we obtain $ \dim_{\F_q}(\mathcal{U}_{i}) \geq m + e $. If $ \dim_{\F_q}(\mathcal{U}_{i}) > m + e $, then, by using Grassmann formula, we derive that $ \dim_{\F_q}(\mathcal{U}_{i} \cap \langle {x}_h \rangle) > e $ for any $h \ne j$, which contradicts \eqref{eq:condUi}. Therefore, we must have $ \dim_{\F_q}(\mathcal{U}_{i}) = m + e $.
    
    \textbf{Case 2}: Suppose that $\mathcal{U}_i$ does not contain any of the subspaces of $\Lambda_1$. By assumption
    \begin{equation}\label{eq:condUie-dim2} \dim_{\fq}(\mathcal{U}_i \cap \langle \mathbf{x}_j\rangle_{\F_{q^m}})=e, \end{equation}
    for any $j$.
    If $e=0$ we get that $\dim_{\fq}(\mathcal{U}_i)=0$, a contradiction with the  nondegenerate condition of the code.
    Suppose that $e>0$ and observe that $n_i\geq 2e$. \eqref{eq:condUie-dim2} also implies that 
    \[ L_{\mathcal{U}_i}=\Lambda_1, \]
    and by combining \eqref{eq:condUie-dim2} and Equations \eqref{eq:pesicard} and \eqref{eq:pesivett}, we have that
    \[ |L_{\mathcal{U}_i}| = \frac{q^{n_i}-1}{q^e-1}. \]
    Hence, we derive that
    \[ q^m+1=|L_{\mathcal{U}_i}| = \frac{q^{n_i}-1}{q^e-1}. \]
    We have that $\frac{q^{n_i}-1}{q^e-1}\ne q^m+1$. Indeed, since $\frac{q^{n_i}-1}{q^e-1} \in \NN$ we have that $e \mid n_i$ and, since $n_i \ne e$, this implies $q^{n_i-e}+\ldots +q^e+1 \equiv q^e+1 \pmod{q^{2e}}$ and  $q^m+1 \equiv 1 \pmod{q^{2e}}$. Hence, these numbers cannot be equal.
    
    Therefore, only Case 1 is possible and so for any $i \in [t]$ we have \eqref{eq:condUi} and $ \dim_{\fq}(\mathcal{U}_{i}) = m + e $, from which it follows that $\mathcal{U}_{i}=\langle {x}_j\rangle_{\F_{q^m}}+ \mathcal{S}_e^i$, where $\mathcal{S}_e^i$ is an $e$-dimensional $\mathbb{F}_q$-subspace of $\mathbb{F}_{q^m}^2$, and $\mathcal{S}_e^i \cap \langle {x}_j \rangle_{\F_{q^m}} = \{0\}$.
    Now, from \eqref{eq:t'andt} we know that for any element in $\Lambda_1$ there are exactly $t'$ elements among the $\mathcal{U}_i$'s through it and there is no $\mathcal{U}_i$ not containing an element of $\Lambda_1$. This implies that the the total number of the $\mathcal{U}_i$'s (counted eventually with multiplicity) is
    \[ t=t'(q^m+1). \]
    The assertion is now proved.
\end{proof}

In the last part of this section, our aim is to provide some insight in the theory of constant rank-profile sum-rank metric codes of larger dimension. Here, we will use again the geometric approach together with a notion of duality, which will allow us to work again with points instead of hyperplanes.

Let $\sigma \colon\F_{q^m}^k\times \F_{q^m}^k\rightarrow \mathbb{F}_{q^m}$ be a  nondegenerate reflexive sesquilinear form over $\F_{q^m}^k$ and consider
$\sigma' \colon \F_{q^m}^k \times \F_{q^m}^k \rightarrow \mathbb{F}_q$ by $\sigma':(u,v)\mapsto \mathrm{Tr}_{q^m/q}(\sigma(u,v))$.
Note that, once we regard $\F_{q^m}^k$ as an $\F_q$-vector space, then $\sigma^\prime$ is again a  nondegenerate reflexive sesquilinear form on $\F_{q^m}^k$.
Let $\perp$ and $\perp'$ be the orthogonal complement maps defined by $\sigma$ and $\sigma'$ on the lattices of $\F_{q^m}$-linear and 
$\F_q$-linear subspaces of $\F_{q^m}^k$, respectively.

We list some properties that we will use later.

\begin{proposition}(see \cite[Section~2]{polverino2010linear})\label{prop:dualityproperties}
With the above notation,
\begin{itemize}
    \item[(i)] $\dim_{\F_{q^m}}(W)+\dim_{\F_{q^m}}(W^\perp)=k$, for every $\F_{q^m}$-subspace $W$ of $\F_{q^m}^k$.
    \item[(ii)] $\dim_{\F_{q}}(U)+\dim_{\F_{q}}(U^{\perp'})=mk$, for every $\F_{q}$-subspace $U$ of $\F_{q^m}^k$.
    \item[(iii)] $T_1\subseteq T_2$ implies  $T_1^{\perp'}\supseteq T_2^{\perp'}$, for every $\F_q$-subspaces $T_1,T_2$ of $\F_{q^m}^k$.
    \item[(iv)] $W^\perp=W^{\perp'}$, for every $\F_{q^m}$-subspace $W$ of $\F_{q^m}^k$.
    \item[(v)] Let $W$ and $U$ be an $\F_{q^m}$-subspace and an $\F_q$-subspace of $\fqm^k$ of dimension $s$ and $t$, respectively. Then
    \begin{equation}\label{eq:dualweight} \dim_{\F_q}(U^{\perp'}\cap W^{\perp'})-\dim_{\F_q}(U\cap W)=mk-t-sm. \end{equation}
\end{itemize}
\end{proposition}

In what follows, we will consider as $\sigma$ the standard scalar product.

\begin{proposition}\label{prop:geometricranklist2}
Let $\mathcal{C}$ be a nondegenerate $(\mu_1,\dots,\mu_t)$-constant rank-profile sum-rank metric code with parameters $[\bfn, k, d]_{q^m/q}$. Let $G = (G_1 \mid \dots \mid G_t)$ be a generator matrix of $\mathcal{C}$. Let $\mathcal{U}_i$ be the $\mathbb{F}_q$-span of the columns of $G_i$, for every $i \in [t]$. 
For all $v \in \mathbb{F}_{q^m}^k$, we have that
\[\tau(m-\dim_{\mathbb{F}_q}(\mathcal{U}_1^{\perp'}\cap \la v\ra_{\F_{q^m}}),\ldots, m-\dim_{\mathbb{F}_q}(\mathcal{U}_t^{\perp'}\cap \la v\ra_{\F_{q^m}}))= (\mu_1,\dots,\mu_t). \]
\end{proposition}
\begin{proof}
From Proposition \ref{prop:geometricranklist}, we have that 
\[\tau(n_1-\dim_{\mathbb{F}_q}(\mathcal{U}_1\cap v^\perp),\ldots, n_t-\dim_{\mathbb{F}_q}(\mathcal{U}_t\cap v^\perp))= (\mu_1,\dots,\mu_t), \]
for any $v \in \F_{q^m}^k\setminus\{0\}$.
Let us fix $v \in \F_{q^m}^k\setminus\{0\}$ and $i \in [t]$ and observe that, by using \eqref{eq:dualweight},
\[ \dim_{\fq}(\mathcal{U}_i\cap v^\perp)=mk-\dim_{\fq}((\mathcal{U}_i\cap v^\perp)^{\perp'})=n_i-m+\dim_{\fq}(\mathcal{U}_i^{\perp'}\cap \la v\ra_{\F_{q^m}}). \]
\end{proof}

\begin{remark}
In other words, we are going to work with a system associated with the geometric dual of the code; see \cite{borello2024geometric}.
\end{remark}

\begin{remark}
    As before, we need to control the intersection of the systems with the hyperplane of $\F_{q^m}^k$.
    The number of $(k-1)$-dimensional $\F_{q^m}$-subspaces in $\F_{q^m}^k$ is $\frac{q^{mk}-1}{q^m-1}$ and we can list them as follows
    \[ \Lambda_{k-1}=\left\{  {x}_1^\perp , \dots, {x}_{\frac{q^{mk}-1}{q^m-1}}^\perp \right\}, \]
    for some ${x}_1,\ldots,{x}_{\frac{q^{mk}-1}{q^m-1}} \in \F_{q^m}^k$.
    Also, note that the set
    \[ \Lambda_{1}=\left\{  \la {x}_1\ra_{\F_{q^m}} , \dots, \la {x}_{\frac{q^{mk}-1}{q^m-1}}\ra_{\F_{q^m}} \right\}, \]
    corresponds to the set of points in $\mathrm{PG}(k-1,q^m)$.
    We will use this notation in the next result.
\end{remark}

We can now prove a characterization result for constant rank-profile sum-rank metric code where the rank profile has all the entries but one equal to $m$.

\begin{theorem}
Let $\mathcal{C}$ be a  nondegenerate $(\underbrace{ m, \dots, m}_{\frac{q^{mk}-1}{q^m-1}-1 \text{ times}}, r)$-constant rank-profile sum-rank metric code with parameters $[\mathbf{n}, k, d]_{q^m/q}$, where $0\leq r< m$ and $n_i=mk-m+r$ for any $i$.
We have that a system associated with $\C$ is $\left(\mathcal{U}_1,\ldots,\mathcal{U}_{\frac{q^{mk}-1}{q^m-1}}\right)$, where
\begin{equation}\label{eq:Uiperp}
    \mathcal{U}_{i} = (\mathcal{S}_{e}^i)^{\perp'},
\end{equation}
where $ \mathcal{S}_{e}^i $ is an $ (m-r) $-dimensional $\fq$-subspace of $\la x_i \ra_{\fqm}$.
Conversely, choosing the $\mathcal{U}_i$'s as in \eqref{eq:Uiperp}, an associated code is a  nondegenerate $(\underbrace{ m, \dots, m}_{\frac{q^{mk}-1}{q^m-1}-1 \text{ times}}, r)$-constant rank-profile sum-rank metric code with parameters $[\mathbf{n}, k, d]_{q^m/q}$.
\end{theorem}
\begin{proof}
    By Proposition \ref{prop:geometricranklist2}, we have that 
    \begin{equation}\label{eq:1erestm} 
    \dim_{\fq}(\mathcal{U}_i^{\perp'}\cap \langle {x}_j \rangle_{\F_{q^m}})=
    \begin{cases}
        0 & \text{for } \frac{q^{mk}-1}{q^m-1} \text{ values of } i , \\
        m-r & \text{for one value of  } i,
    \end{cases} \end{equation}
    for any $j \in \left[ \frac{q^{mk}-1}{q^m-1} \right]$.
    Since the number of the $\mathcal{U}_i$'s is exactly has the size of $\Lambda_{k-1}$ and since $\dim_{\fq}(\mathcal{U}_i^{\perp'})=m-r$, we have that, up to reordering the $\mathcal{U}_i$'s, for any $i \in \left[\frac{q^{mk}-1}{q^m-1}\right]$ Equation \eqref{eq:Uiperp} is satisfied.
    For the converse, one can immediately check the parameters by using again Proposition \ref{prop:geometricranklist2}.
\end{proof}

\section{One-weight MSRD codes}

In this section, we will discuss one-weight codes that are also maximum sum-rank distance codes. We will survey on the case of codes of dimension two and we will discuss the dimension three case.

We start by recalling the geometric description of an MSRD code.

\begin{corollary}(see \cite[Corollary 3.10]{neri2023geometry})\label{teo:designMSRD} 
Let $\C$ be a sum-rank nondegenerate $\Fmnkd$ code and let $(\mathcal{U}_1,\ldots,\mathcal{U}_t)$ be a system associated with $\C$. $\mathcal{C}$ is an MSRD code if and only if 
\[
\max\left\{ \sum_{i=1}^t \dim_{\fq}(\mathcal{U}_i \cap \mathcal{H})  \st \mathcal{H} \mbox{ hyperplane of }\F_{q^m}^k  \right\} \leq k-1.
\]
\end{corollary}

Therefore, a one-weight MSRD code can be characterized geometrically as follows.

\begin{corollary}\label{teo:designMSRDoneweight} 
Let $\C$ be a sum-rank nondegenerate $\Fmnkd$ code and let $(\mathcal{U}_1,\ldots,\mathcal{U}_t)$ be a system associated with $\C$. $\mathcal{C}$ is a one-weight MSRD code if and only if 
\[
 \sum_{i=1}^t \dim_{\fq}(\mathcal{U}_i \cap \mathcal{H})   = k-1,
\]
for every hyperplane $\mathcal{H}$ of $\F_{q^m}^k$.
\end{corollary}

An important property that is satisfied by the systems associated with MSRD code is the following.

\begin{proposition}\label{prop:boxsdesign}(see \cite[Proposition 3.2]{santonastaso2023subspace})
    Let $\C$ be an MSRD code with parameters $\Fmnkd$ and let $(\mathcal{U}_1,\ldots,\mathcal{U}_t)$ be a system associated with $\C$.
    We have that 
    \[
 \sum_{i=1}^t \dim_{\fq}(\mathcal{U}_i \cap \mathcal{S})   \leq j,
\]
for every $j$-dimensional $\fqm$-subspace $\mathcal{S}$ of $\F_{q^m}^k$.
\end{proposition}

This implies that if $(\mathcal{U}_1,\ldots,\mathcal{U}_t)$ is a system associated with an MSRD code, then the linear sets $L_{\mathcal{U}_i}$'s are scattered with respect to the hyperplanes (i.e. $w_{L_{\mathcal{U}_i}}(\mathcal{H})\leq k-1$ for any hyperplane in $\mathrm{PG}(k-1,q^m)$ and any $i \in [t]$, see \cite{lunardon2017mrd,sheekey2019scatterd}) and they are pairwise disjoint.
We will use these properties to handle the cases where the dimension of the code is either two or three.
One-weight MSRD codes of dimension two have been studied in \cite{neri2023geometry}.
By using Corollary \ref{teo:designMSRDoneweight} one can prove the following.

\begin{proposition}(see \cite[Corollary 7.2]{neri2023geometry})
    Let $\C$ be a sum-rank nondegenerate $\Fmnkd$ code and let $(\mathcal{U}_1,\ldots,\mathcal{U}_t)$ be a system associated with $\C$. $\mathcal{C}$ is a one-weight MSRD code if and only if for every $i \in [t]$, $L_{\mathcal{U}_i}$ is a scattered $\fq$-linear set in $\mathrm{PG}(1,q^m)$, the $L_{\mathcal{U}_i}$'s are pairwise disjoint and $L_{\mathcal{U}_1}\cup\ldots\cup L_{\mathcal{U}_t}=\mathrm{PG}(1,q^m)$.
\end{proposition}

In \cite[Corollary 7.2]{neri2023geometry}, it has been proved that for a two-dimensional one-weight MSRD code with $t$ blocks we have that
\[q+1\leq t \leq q^m+1.\]
Moreover, constructions of two-dimensional one-weight MSRD codes have been shown in the cases when $t=q+1$, $t=q^m+1$ and some other cases; see \cite[Section 7.1]{neri2023geometry}.

\subsection{Dimension three}

In this section, we will study the case of one-weight MSRD codes with dimension three, where all the blocks have length equal to $m$. We will first derive some conditions that the parameters of a one-weight MSRD code must satisfy, and then we will derive some non-existence results.

We start by determining the size of the intersection of the lines of $\mathrm{PG}(2,q^m)$ and the union of the linear sets associated with the code.

\begin{proposition}\label{prop:linek=3}
    Let $\C$ be a one-weight MSRD code with parameters $[\bfm,3,d]_{q^m/q}$ with $\bfm=(m,\ldots,m)$, and let $(\mathcal{U}_1,\ldots,\mathcal{U}_t)$ be a system associated with $\C$.
    For every line $\ell$ in $\mathrm{PG}(2,q^m)$ we have that one of the following holds
    \begin{itemize}
        \item there exists $i \in [t]$ such that $|\ell \cap L_{\mathcal{U}_i}|=q+1$ and $\ell \cap L_{\mathcal{U}_j}=\emptyset$ for any $j \ne i$;
        \item there exist $i_1,i_2 \in [t]$ such that $|\ell \cap L_{\mathcal{U}_{i_1}}|=|\ell \cap L_{\mathcal{U}_{i_2}}|=1$ and $\ell \cap L_{\mathcal{U}_j}=\emptyset$ for any $j \ne i_1,i_2$.
    \end{itemize}
\end{proposition}
\begin{proof}
    By Corollary \ref{teo:designMSRDoneweight}, we have that
    \[\sum_{i \in [t]} w_{L_{\mathcal{U}_i}}(\ell)=2,\]
    for any line $\ell$. 
    Let us fix a line $\ell$, we have that either there exists $i \in [t]$ such that $w_{L_{\mathcal{U}_i}}(\ell)=2$ and $w_{L_{\mathcal{U}_j}}(\ell)=0$ for every $j\ne i$ or there exist $i_1,i_2 \in [t]$ such that  $w_{L_{\mathcal{U}_{i_1}}}(\ell)=w_{L_{\mathcal{U}_{i_2}}}(\ell)=1$ and $w_{L_{\mathcal{U}_j}}(\ell)=0$ for every $j\ne i_1,i_2$.
    By Proposition \ref{prop:boxsdesign}, the $L_{\mathcal{U}_i}$'s are scattered and so $w_{L_{\mathcal{U}_i}}(\ell)=2$ implies that $|L_{\mathcal{U}_i}\cap \ell|=q+1$ and the assertion follows.
\end{proof}

This implies that the union of the linear sets associated with a system as in the above proposition is a $2$-fold blocking set.
A $2$-\textbf{fold blocking set} in $\mathrm{PG}(2,q)$ is a set of points $B$ in $\mathrm{PG}(2,q)$ such that every line intersects $B$ in at least two distinct points. We say that it is \textbf{minimal} is it does not contain a proper subset which is a $2$-fold blocking set, or equivalently if for each point there exists a $2$-secant line (i.e. a line intersecting $B$ in exaclty two distinct points). We refer to \cite{blokhuis2011blocking} for a general overview on blocking sets.

\begin{proposition}\label{prop:minimal2bs}
    Let $\C$ be a sum-rank nondegenerate $[\bfm,3,d]_{q^m/q}$ one-weight MSRD code with $\bfm=(m,\ldots,m)$, and let $(\mathcal{U}_1,\ldots,\mathcal{U}_t)$ be a system associated with $\C$.
    The set $L_{\mathcal{U}_1}\cup \ldots \cup L_{\mathcal{U}_t}$ is a minimal $2$-fold blocking set of size $t\frac{q^m-1}{q-1}$.
\end{proposition}
\begin{proof}
    By Proposition \ref{prop:linek=3}, every line intersects $L_{\mathcal{U}_1}\cup \ldots \cup L_{\mathcal{U}_t}$ in either two or $q+1$ points and so $L_{\mathcal{U}_1}\cup \ldots \cup L_{\mathcal{U}_t}$ is a $2$-fold blocking set.
    The minimality is due to the fact that if $P \in L_{\mathcal{U}_1}\cup \ldots \cup L_{\mathcal{U}_t}$ then there exists only one $i \in [t]$ such that $P \in L_{\mathcal{U}_i}$ (as the $ L_{\mathcal{U}_i}$'s are pairwise disjoint).
    Therefore, let $Q \in (L_{\mathcal{U}_1}\cup \ldots \cup L_{\mathcal{U}_t})\setminus L_{\mathcal{U}_i}$, by Proposition \ref{prop:linek=3}, the line $PQ$ is a $2$-secant line. Hence, $L_{\mathcal{U}_1}\cup \ldots \cup L_{\mathcal{U}_t}$ is a minimal $2$-fold blocking set.
\end{proof}

Due to the behavior of the lines with respect to the linear sets of a system associated with a one-weight MSRD code, a relation involving $q,m$ and $t$ needs to be satisfied.

\begin{proposition}\label{prop:condont}
    Let $\C$ be a one-weight MSRD code with parameters $[\bfm,3,d]_{q^m/q}$  with $\bfm=(m,\ldots,m)$. We have that
    \begin{equation}\label{eq:condontk=3} t\frac{\frac{q^m-1}{q-1}\left(\frac{q^m-1}{q-1}-1\right)}{{q+1 \choose 2}} + {t \choose 2}\left( \frac{q^m-1}{q-1} \right)^2=q^{2m}+q^m +1.\end{equation}
    In particular, $t$ is a root of the $Ax^2+Bx+C$, where 
    \[A=\frac{1}{2} \left( \frac{q^m - 1}{q - 1} \right)^2,\quad  B= \frac{(q^m - 1)}{2(q - 1)^2(q + 1)} \left( 2q^{m-1} - q^{m+1} - q^m + q - 1 \right),\quad  C=-\left( q^{2m} + q^m + 1 \right).\]
\end{proposition}
\begin{proof}
    Let $(\mathcal{U}_1,\ldots,\mathcal{U}_t)$ be a system associated with $\C$.
    By Proposition \ref{prop:linek=3}, we can consider the following sets of lines in $\mathrm{PG}(2,q^m)$:
    \begin{itemize}
        \item $\mathcal{L}_1$, which is the set of lines meeting one of the $L_{\mathcal{U}_i}$'s in $q+1$ points;
        \item $\mathcal{L}_2$, which is the set of lines meeting exactly two $L_{\mathcal{U}_i}$'s into two distinct points.
    \end{itemize}
    Clearly, $\mathcal{L}_1\cap\mathcal{L}_2=\emptyset$ and $\mathcal{L}_1\cup\mathcal{L}_2$ corresponds to the set of all lines in $\mathrm{PG}(2,q^m)$.
    The number of lines that meet $L_{\mathcal{U}_1}$ is $q+1$ points is the number of pairs of distinct points in $L_{\mathcal{U}_1}$ divided by the number of pairs that give a fixed line, i.e.
     
    \[\frac{\frac{q^m-1}{2(q-1)}\left(\frac{q^m-1}{q-1}-1\right)}{{q+1 \choose 2}},\]
    and the same happens for all the other linear sets, therefore \[ |\mathcal{L}_1|=t\frac{\frac{q^m-1}{2(q-1)}\left(\frac{q^m-1}{q-1}-1\right)}{{q+1 \choose 2}}. \]
    For the size of $\mathcal{L}_2$ we need to count the number of pairs of points belonging to distinct $L_{\mathcal{U}_i}$'s. Let $i_1,i_2 \in [t]$ with $i_1 \ne i_2$, the number of pairs of points one belonging to $L_{\mathcal{U}_{i_1}}$ and one belonging to $L_{\mathcal{U}_{i_2}}$ is 
    \[\left( \frac{q^m-1}{q-1} \right)^2.\]
    The size of $\mathcal{L}_2$ is then obtained by multiplying this number by the number of choices of two distinct elements in the et $[t]$, so
    \[ |\mathcal{L}_2|={t \choose 2}\left( \frac{q^m-1}{q-1} \right)^2. \]
    Therefore, since $|\mathcal{L}_1\cup \mathcal{L}_2|=|\mathcal{L}_1|+| \mathcal{L}_2|=q^{2m}+q^m+1$, we obtain \eqref{eq:condontk=3}.

    Moreover, observe that \eqref{eq:condontk=3} can be written as follows
\[
\frac{\frac{t}{2}\left(\frac{q^m - 1}{q - 1}\right)\left(\frac{q^m - 1}{q - 1} - 1\right)}{\binom{q + 1}{2}} + \left(\frac{q^m - 1}{q - 1} \right)^2 \frac{t(t-1)}2-( q^{2m} + q^m + 1)=0,
\]
thus, we obtain
\[
\frac{t^2}{2} \left( \frac{q^m - 1}{q - 1} \right)^2
+ \frac{t (q^m - 1)}{2(q - 1)^2(q + 1)} \left( 2q^{m-1} - q^{m+1} - q^m + q - 1\right)
- \left( q^{2m} + q^m + 1 \right) = 0
,\]
and so 
\[At^2+Bt+C=0.\]
\end{proof}

In order to give some bounds on the possible values of $t$ we need the following technical lemma.

\begin{lemma}\label{bounds}
Let $m\geq 3$ and $q\geq 16$.
Let
\[A=\frac{1}{2} \left( \frac{q^m - 1}{q - 1} \right)^2,\quad B=  \frac{ (q^m - 1)}{2(q - 1)^2(q + 1)} \left( 2q^{m-1} - q^{m+1} - q^m + q - 1\right),\quad  C=-\left( q^{2m} + q^m + 1 \right),\]
and \(\Delta=B^2-4AC\), we have that
\[\frac{q^{m+1}(q^m-1)}{2(q-1)^2(q+1)}\leq -B\leq \frac{q^{m+2}(q^m-1)}{2(q-1)^2(q+1)}, \]
and
\[\frac{7q^{2m+4}(q^m-1)^2}{4(q-1)^4(q+1)^2}\leq \Delta \leq \frac{q^{2m+5}(q^m-1)^2}{4(q-1)^4(q+1)^2}. \]
\end{lemma}
    \begin{proof}
     \textbf{Bounds on \(B\)}.\\
     We rewrite \(B\) as  \[  -B= \frac{q^m-1}{2(q-1)^2(q+1)} \left(q^{m+1}+q^m-2q^{m-1}-q+1\right). \]
     since \[{q^{m+1}+q^m-2q^{m-1}-q+1}\geq q^{m+1},\]
     we have that
     \[-B\geq \frac{q^{m+1}(q^m - 1)}{2(q - 1)^2(q + 1)}.\]
     Similarly, since 
     \[ q^{m+1}+q^m-2q^{m-1}-q+1\leq q^{m+2},\]
     we get
     \[-B\leq \frac{q^{m+2}(q^m - 1)}{2(q - 1)^2(q + 1)}.\]
     Thus, the bounds on \(-B\) are established.\\
     \textbf{Bounds on \(\Delta\)}.\\
     We expand \(\Delta\) using \(A, B, \text{ and } C\)
     \[
\Delta=
\frac{(q^m - 1)^2}{4(q - 1)^4(q + 1)^2}
\left[
\left(2q^{m-1}
-q^{m+1}-q^m+q-1
\right)^2
+
8(q - 1)^2(q + 1)^2
\left(q^{2m}+q^m+1\right)
\right].
\]
To lower bound \(\Delta\), note that 
\[8(q - 1)^2(q + 1)^2(q^{2m} + q^m + 1)=8(q^4+1-2q^2)(q^{2m}+q^m+1).\] 
Observe that $8q^{2m+4}\geq 7q^{2m+4}$ and $8(q^4+1-2q^2)(q^{2m}+q^m+1)-7q^{2m+4}\geq 0$, it follows that \[
8(q - 1)^2(q + 1)^2(q^{2m} + q^m + 1) \geq  7q^{2m+4}.
\]
Therefore
\[\Delta\geq \frac{7q^{2m+4}(q^m - 1)^2}{4(q - 1)^4(q + 1)^2}.\]
To upper bound \(\Delta\), we expand
\[8(q^4-2q^2+1)(q^{2m} + q^m + 1)=8q^{2m+4}+8q^{m+4}+8q^4-16q^{2m+2}-16q^{m+2}-16q^2+8q^{2m}+8q^m+8.\]
Hence
\[8q^{2m+4}+8q^{m+4}+8q^4+8q^{2m}+8q^m+8\leq \frac
{q^{2m+5}}{2}+ 16q^{2m+2}+16q^{m+2}+16q^2 ,\]
which follows from the following inequalities 
\begin{align*}
8q^{2m+4} &\leq \frac{q^{2m+5}}{2}, \\
8q^{m+4} + 8q^{2m} &\leq 16q^{2m+2}, \\
8q^m + 8q^4 &\leq 16q^{m+2}, \\
8 &\leq 16q^2.
\end{align*}
Moreover
 \[ \left( 2q^{m-1} -q^{m+1}-q^m+q-1 \right)^2 \leq \frac{q^{2m+5}}{2}, \]
since \( 2\left( 2q^{-1} -q-1+q^{-m+1} -q^{-m} \right)^2 \leq \frac{q^5}{2}\). Therefore,
\[\Delta\leq\frac{q^{2m+5}(q^m-1)^2}{4(q-1)^4(q+1)^2}.\]
This completes the proof.
    \end{proof}
Thanks to the above results, we can obtain the number of blocks necessary to have an MSRD codes with parameters $[(m,\dots,m),3,tm-2]_{q^m/q}$.

\begin{theorem}\label{thm:maink=3}
Suppose there exists a one-weight MSRD code with parameters $[(m,\dots,m),3,tm-2]_{q^m/q}$ then 
\[t=\frac{-B+\sqrt{\Delta}}{2A},\]
where \[A=\frac{1}{2} \left( \frac{q^m - 1}{q - 1} \right)^2,\quad B=  \frac{ (q^m - 1)}{2(q - 1)^2(q + 1)} \left(2q^{m-1} - q^{m+1} - q^m + q - 1\right),\quad  C=-\left( q^{2m} + q^m + 1 \right),\] and\[\Delta=B^2-4AC.\]\\
In particular, if $m\geq 3$ and $q\geq 16$ we have that $q\leq t \leq q^{3/2}$.
\end{theorem}
\begin{proof}
By Proposition \ref{prop:condont}, we know that
\[
t = \frac{-B \pm \sqrt{\Delta}}{2A}.
\]
First, it is easy to see that $-B - \sqrt{\Delta} < 0$ and \(A>0\). Since $t \in \NN$, this implies that $t = \frac{-B + \sqrt{\Delta}}{2A}$.

Now, applying Lemma \ref{bounds}, we derive upper bound for \(q\geq16\), and \(m\geq3\)
 \[
t \leq \left( \frac{q^{m+2}(q^m - 1)}{2(q - 1)^2(q + 1)} + \frac{q^{m + \frac{5}{2}}(q^m - 1)}{2(q - 1)^2(q + 1)} \right) \cdot \frac{(q - 1)^2}{(q^m - 1)^2},
\]
consequently, by straightforward computation, we obtain
\[t\leq \frac{q^{m+2}}{2(q^m-1)(q+1)}\left(1+q^{1/2} \right),\]
so
\[ \frac{q^{m+2}}{2(q^m-1)(q+1)}\left(1+q^{1/2} \right)\leq q^{3/2},\]
which is equivalent to 
\[1+q^{1/2}\leq 2q^{-1/2-m}(q^m-1)(q+1),\]
this holds for all \(q\geq16\) and \(m\geq 3\).\\
Similarly, after applying Lemma \ref{bounds} to obtain a lower bound, we have
\[t\geq \frac{q^{m+1}}{2(q^m-1)(q+1)}\left(1+\sqrt{7}q\right),\]
straightforward computations show that 
\[t\geq \frac{q^{m+1}}{2(q^m-1)(q+1)}\left(1+\sqrt{7}q\right)\geq q.\]
 Thus, we have proved the desired result.
 \end{proof}

We now consider the special case where \(q=2\) in the following theorem, which is excluded from the previous considerations.

\begin{theorem}
    For any $m\in \NN$, there do not exist one-weight MSRD code with parameters \([(m,\dots,m),3,tm-2]_{2^m/2}\). 
\end{theorem}
    \begin{proof}
    Suppose there exists one-weight MSRD code $\C$ with parameters \([(m,\dots,m),3,2t-2]_{2^m/2}\). By Theorem \ref{thm:maink=3}, we know that $t=\frac{-B+\sqrt{\Delta}}{2A}$,
    where
    \[A=\frac{(2^m-1)^2}{2},\quad  B=-\frac{(2^m-1)(2^{m+1}-1)}{6},\quad C=-(2^{2m}+2^m+1).\]
    Accordingly, the discriminant \(\Delta\) becomes
      \[\Delta= \frac{(2^m-1)^2}{36} \left( 76\cdot2^{2m}+68\cdot2^m+73 \right).\]
      Consequently, the value of \(t\)
      \[ t= \frac{ 2^{m+1}-1+ \sqrt{76\cdot2^{2m}+68\cdot2^m+73} }{ 6(2^m-1) }  <2, \quad \text{for all } m\geq4. \]
      Therefore, the only possibility is that $t=1$ (as it needs to be a positive integer).
      Hence, $t=1$, which is a contradiction to the nonexistence of a nontrivial one-weight code in the rank metric; see \cite{loidreau2006properties}.
    \end{proof}

For the case $m=2$, we are going to use the fact that the union of the linear set associated to the system is a $2$-fold blocking set and the following result.

\begin{theorem}[see \text{\cite[Theorem 5.1]{ball1996size}}]\label{thm:BBbound}
     Suppose that $B$ is a $2$-fold blocking set in \(\text{PG}(2,q^2)\). Then \(|B|\geq 2q^2+2q+2.\)
\end{theorem}

Combining Theorem \ref{thm:BBbound} and Proposition \ref{prop:condont} we obtain the following nonexistence result for $m=2$.

\begin{corollary}
      For any prime power $q$, there do not exist one-weight MSRD code with parameters \([(2,\dots,2),3,2t-2]_{q^2/q}\).
\end{corollary}
     \begin{proof}
        Suppose there exists one-weight MSRD code $\C$ with parameters \([(2,\dots,2),3,2t-2]_{q^2/q}\).
        By Theorem \ref{thm:maink=3}, we know that $t=\frac{-B+\sqrt{\Delta}}{2A}$,
    where 
         \[ A=\frac{(q+1)^2}{2},\quad B=\frac{-q^2-2q+1}{2}, \quad C=-(q^4+q^2+1). \]
         The discriminant \(\Delta\) becomes
          \[ \Delta=\frac{1}{4}\left( 8q^6+16q^5+17q^4+20q^3 +18q^2+12q+9 \right). \]
         Thus, by Proposition \ref{prop:condont}, we obtain \(t\)
          \[ t=\frac{ q^2+2q-1 +\sqrt{ 8q^6+16q^5+17q^4+20q^3 +18q^2+12q+9 }} {2(q+1)^2}. \]
         From this we can immediately derive that
         \[  t\leq \sqrt{2}q.\]
         Let $(U_1,\ldots,U_t)$ be an associated system to $\C$. By Proposition \ref{prop:minimal2bs}, $L_{U_1}\cup \ldots \cup L_{U_t}$ is a minimal $2$-fold blocking set.
         By Theorem \ref{thm:BBbound}, we have that
         \[|L_{U_1}\cup \ldots \cup L_{U_t}|\geq 2q^2+2q+2.\]
         Now, since the $L_{U_i}$'s are scattered $\fq$-linear sets of rank $2$ pairwise disjoint, we obtain that 
         \[|L_{U_1}\cup \ldots \cup L_{U_t}|=t(q+1)\leq \sqrt2q(q+1),\]
         as $t\leq \sqrt2 q$.
         We get a contradiction as $\sqrt2q(q+1)$ is strictly less than $2q^2+2q+2$.
         \end{proof}

\section*{Acknowledgements}
We thank Alessandro Neri for pointing out a correction to Section 4 and Luca Lombardi for pointing out an oversight in Section 6.
The research  was partially supported by the Italian National Group for Algebraic and Geometric Structures and their Applications (GNSAGA - INdAM).
\bibliographystyle{abbrv}
\bibliography{Biblo}
\end{document}